\documentclass[11pt]{amsart}
\usepackage{amsmath}
\usepackage{amsfonts}
\usepackage{xcolor} 
\usepackage{epsfig}
\usepackage{mathtools}
\usepackage{amsthm}
\usepackage{mathabx}
\usepackage{cite}
\usepackage{cancel}

\newcommand{\comment}[1]{}

\numberwithin{equation}{section}

\newcounter{mnotecount}

\newcommand{\mnotex}[1]
{\protect{\stepcounter{mnotecount}}$^{\mbox{\footnotesize $\bullet$\themnotecount}}$ 
\marginpar{
\raggedright\tiny\em
$\!\!\!\!\!\!\,\bullet$\themnotecount: #1} }

\newcommand{\ce}{\mathcal{E}}
\newcommand{\cV}{\mathcal{V}}
\newcommand{\bg}{\boldsymbol{g}}
\newcommand{\si}{\sigma}
\newcommand{\cc}{\textbf{c}}
\newcommand{\cT}{\mathcal{T}}
\newcommand{\ID}{\mathrm{ID}}

\begin{document}

\newtheorem{theorem}{Theorem}
\newtheorem{proposition}{Proposition}
\newtheorem{definition}{Definition}
\newtheorem{example}{Example}
\newtheorem{remark}{Remark}

\title[Higher fundamental forms of asymptotically de Sitter spacetimes]{Higher fundamental forms of the conformal boundary of asymptotically de Sitter spacetimes}
                     
\author[A.\ R. Gover]{A. Rod Gover}
\address{Department of Mathematics, The University of Auckland, Private Bag 92019, Auckland 1142, New Zealand}
\email{r.gover@auckland.ac.nz}

\author[J. Kopi\'nski]{Jaros\l aw Kopi\'nski}
\address{Center for Theoretical Physics, Polish Academy of Sciences,
Warsaw, Poland}
\email{jkopinski@cft.edu.pl}

\begin{abstract}
We provide a partial characterization of the conformal infinity of
asymptotically de Sitter spacetimes by deriving constraints that
relate the asymptotics of the stress-energy tensor with conformal
geometric data. The latter is captured using recently defined objects,
called higher conformal fundamental forms. For the boundary
hypersurface, these generalize to higher order the trace-free part of
the second form.
\end{abstract}

\maketitle

\section{Introduction}
Spacetimes with  positive cosmological constant $\Lambda$ have
attracted increasing attention in mathematical relativity in recent
years, see e.g. \cite{a1,a2, senovilla}. This was motivated by the
observational implications that the Universe is best described if
$\Lambda > 0$ is included in the Einstein field equations
\cite{lambda}. Moreover, asymptotically de Sitter spacetimes are also
used in the context of the dS/CFT correspondence \cite{strominger} and
in the Conformal Cyclic Cosmology scenario \cite{ccc}, which relies on
the positivity of the cosmological constant.

The natural way to study the asymptotic structure of a spacetime
$(\widetilde{M}, \widetilde{g}_{ab})$ is through the conformal
Einstein field equations formalism, introduced in \cite{cefe}. In this
approach, one considers a conformal extension (unphysical spacetime)
$\left(M, \mathtt{g}_{ab}\right)$ of $(\widetilde{M},
\widetilde{g}_{ab})$ -- a four-dimensional manifold with the boundary
$\Sigma$ such that $\widetilde{M}$ can be identified with the interior
of $M$ and the metric $\widetilde{g}_{ab}$ is singular on
$\Sigma$. Certain global problems associated with the solution
$(\widetilde{M}, \widetilde{g}_{ab} )$ of the Einstein field equations
can then be studied in terms of the local analysis in the
neighbourhood of $\Sigma$.

Under appropriate conditions regarding the conformal rescaling of the
matter fields one can show that the conformal Einstein field equations
form a regular system of partial differential equations on the
unphysical spacetime $\left(M, \mathtt{g}_{ab} \right)$, see
e.g. \cite{cefemax} for the case where gravitational field is coupled
to the Maxwell and Yang-Mills fields. It should be noted that if a
non-vanishing cosmological constant is included in the system, the
conformal boundary is an umbilic hypersurface. This is a prime example
of the constraints on the conformal fundamental forms of $\Sigma$, as
such hypersurface has vanishing trace-free extrinsic curvature.

The other approach to describe asymptotically de Sitter spacetimes is
based on the Fefferman-Graham power series expansion of the metric,
which stems from the theory of conformal invariants of manifolds
of arbitrary dimensions \cite{fg}. The same type of construction
appeared earlier in a somewhat similar setting of solving Einstein
field equations with matter fields and a positive cosmological
constant \cite{starobinsky}. The key feature of this approach is that
the conformally rescaled bulk metric and the physical stress-energy
tensor $\widetilde{T}_{ab}$ are expanded in terms of the geodesic
distance to the conformal boundary and the Einstein field equations
are solved order by order. An example of such procedure is given in
\cite{tod} and \cite{nurek} in the construction of asymptotically de
Sitter aeons in the Conformal Cyclic Cosmology model.

Our work here utilizes the tractor calculus in the study of conformal
infinity of spacetimes with positive cosmological constant. It is an
efficient and effective tool for studying conformal invariants and
invariant operators in conformal geometry. The natural connection
between tractor calculus and general relativity stems from the fact
that in the study of asymptotic properties of spacetimes one considers
their conformal extension, which focuses on the causal structure while
abandoning the notion of distance. An extended discussion of the
tractor calculus can be found in a review \cite{Curry-G}, while its
application in general relativity can be found for example in
\cite{CGmass,herfray, wally,Borthwick,LT}, among other works.

The main result of this paper builds on the notion of conformal
fundamental forms introduced in \cite{fundforms} and aims to utilize
it in the setting of asymptotically de Sitter spacetimes with the
stress-energy tensor $\widetilde{T}_{ab}$. We will focus on the most
commonly assumed decay rates of $\widetilde{T}_{ab}$, and formulate
the main result in terms of constraints that relate matter fields
along the conformal boundary to its conformal fundamental forms
$\mathring{\mathcal{K}}^{(d)}_{ab}$, $d=2,3,4,5$. It turns out that 
this amounts to showing how the divergence of suitable  projected part of the 
Cotton, Weyl, and Bach tensors on
$\Sigma$ are related to the asymptotics of the stress-energy tensor there.

It should be noted that, for a given decay rate of matter fields, the
definition of conformal fundamental forms given in \cite{fundforms}
can only be applied directly to derive the form of a certain number of
those objects in terms of the intrinsic geometry of the conformal
boundary alone -- a contribution from $\widetilde{T}_{ab}$ would appear
in them, when na\"{i}vely using the definition beyond this range. In
principle, the theory could be extended to generate such objects to arbitrarily   high order, and for any decay rate of matter fields.  Here we avoid
these difficulties and our conformal fundamental form will precisely
match the ones from \cite{fundforms}, modulo the sign of the norm of
the normal vector to the conformal boundary and the asymptotic value
of the scalar curvature (determined by the cosmological
constant). This aligns well with the conformal treatment of
asymptotically de Sitter spacetimes widely used in the literature and
reveals conformal fundamental forms simply linked to the well-known
Weyl, Cotton and Bach tensors.

The main result of this article has the form of the following theorem.
\begin{theorem} \label{mth}
Let $(\widetilde{M}, \widetilde{g}_{ab} )$ be a four-dimensional asymptotically de Sitter spacetime with the stress-energy tensor $\widetilde{T}_{ab}$ and the positive cosmological constant $\Lambda$ admitting a conformal extension $(M, \mathtt{g}_{ab} )$ with
\begin{equation}
\widetilde{g}_{ab} = \Omega^{-2} \mathtt{g}_{ab}, \quad  \widetilde{T}_{ab} = \Omega^q \mathtt{T}_{ab}, \quad q  \in \{0,1,2 \},
\end{equation}
where $\mathtt{T}_{ab}$ is the (regular everywhere) unphysical stress-energy tensor. Then, 
\begin{align} 
& \left(q-2 \right) n^b \mathtt{T}_{a b} +  n_a  \mathtt{T} = 0 \quad \mathrm{on} \quad \Sigma, \label{mthconst1} \\
& \nabla_n \left[ \left(q-2 \right) \mathtt{T}_{ab} n^b +  n_a \mathtt{T} \right] -  \nabla_b \mathtt{T}_a{}^b = 0 \quad \mathrm{on} \quad \Sigma,
\end{align}
where
\begin{equation}
n_a n^a = -1 + \frac{2   \mathtt{H}}{\sqrt{3 \Lambda}} \Omega + \mathcal{O}\left( \Omega^2 \right),
\end{equation}
$\mathtt{H}$ is the mean curvature of $\Sigma$, and $\mathtt{T} := \mathtt{g}^{cd } \mathtt{T}_{cd}$. Moreover, for $q=0, 2$ the $q+3$ conformal fundamental form $\mathring{\mathcal{K}}^{(q+3)}_{ab}$ of the conformal infinity $\Sigma$ is related to the intrinsic trace-free part of $\mathtt{T}_{ab} $, i.e.

\begin{equation} \label{mtheq}
\mathring{\mathcal{K}}^{(q+3)}_{ab} = C \left( q, \Lambda \right) \mathring{\top} \left( \mathtt{T}_{ab} \right) \quad \mathrm{on} \quad \Sigma,
\end{equation}
where ($\mathring{\top}$) $\top$ denotes the (trace-free) projection on $\Sigma$. If $q \geq 1$, then the divergence of the fourth conformal fundamental form $\mathring{\mathcal{K}}^{(4)}_{ab}$ of $\Sigma$ satisfies the following constraint,
\begin{equation}
\begin{split}
\overline{\nabla}^b \mathring{\mathcal{K}}^{(4)}_{ab} =
\begin{cases}
\left( \frac{\Lambda}{3} \right)^{3/2}  \left( \frac{1}{3} \overline{\nabla}_a \mathtt{T}  - \sqrt{\frac{\Lambda}{3}} j_a  \right) \quad \mathrm{for} \quad q=1, \\
- \frac{\Lambda^2}{9}  \top \left(  n^b \mathtt{T}_{ab}\right)  \quad \mathrm{for} \quad q=2
\end{cases}
\end{split}
\end{equation}
on $\Sigma$, where $\overline{\nabla}_a$ is the induced Levi-Civita
connection on $\Sigma$ and $j_a$ is defined by the expansion
$\mathtt{T}_{a b} n^b = \Omega j_a + \mathcal{O}(\Omega)$. The
$\mathring{\mathcal{K}}^{(4)}_{ab}$ is otherwise undetermined by the
local data on the conformal boundary $\Sigma$ but a constraint
\begin{equation} \label{mthconst2}
\mathtt{T}_{ab} = - n_a n_b \mathtt{T}\quad \mathrm{on} \quad \Sigma
\quad \mathrm{for} \quad q=1
\end{equation}
arises.
\end{theorem}
It can be shown that the condition for the regularity of the conformal
field equations from \cite{friedrich} is an example of applying
(\ref{mthconst1}) in the $q=0$ case. The constraint (\ref{mthconst2})
has been previously derived with the use of different methods (see
e.g. \cite{a2} and \cite{senovilla}). The fourth fundamental form of
the conformal boundary cannot be determined locally by the
stress-energy tensor, because (as we shall see) it is an image of the
Dirichlet-to-Neumann map for the conformal Einstein field equations
(viewed as boundary value problem) -- in the spirit of
\cite{GrahamDN}.

\begin{remark}
The principal parts of the higher-order conformal fundamental forms
discussed in this paper are given by the electric part of the ambient Weyl
tensor $C_{a c b d} n^c n^d$ and the projected part of the Cotton $
\mathring{\top} \left( A_{a c b} n^c \right)$ and Bach $
\mathring{\top} \left( B_{ab} \right)$ tensors.
\end{remark}

Strictly speaking,  a spacetime $(\widetilde{M}, \widetilde{g}_{ab}
)$, as in Theorem \ref{mth}, should not be called
asymptotically de Sitter in the case when $q=0$. Nevertheless, we include this case here as
it naturally fits in the current application of tractor calculus to
general relativity.

The structure of this article is as follows. In Section \ref{secprel}
we discuss the definitions of asymptotically de Sitter spacetimes and
their conformal extensions together with some basic concepts from the
tractor calculus. The following section is then dedicated to the application of
tractor calculus in the setting of asymptotically de Sitter
spacetimes. Section \ref{fundformsec} contains the derivation of
constraints relating the matter fields to the conformal fundamental
forms of the conformal boundary of asymptotically de Sitter
spacetimes. The final section contains a discussion about a possible
application of the results of this work.

\subsection*{Notation and conventions}
We will work with an $n$-dimensional ($n\geq 3$) conformal manifold
equipped with the equivalence class of metric tensors $\mathbf{c}$ of
signature $\left(1, n-1,0 \right)$, and then restrict our attention to
$n=4$ in Section \ref{fundformsec}. The abstract index notation will
be used throughout the paper with the lower case Latin letters $a,b,c
\dots$ associated with tensors, and upper case ones $A,B,C \dots$ with
tractors.

Any tensor $T_{ab}$ can be decomposed into its symmetric and antisymmetric parts in accordance with the formula
\begin{equation}
T_{ab } = T_{(ab) } + T_{[ab] },
\end{equation}
where 
\begin{equation}
 T_{(ab) } := \frac{1}{2} \left( T_{ab } + T_{ba } \right), \quad  T_{[ab] } := \frac{1}{2} \left( T_{ab} - T_{ba} \right),
\end{equation}
The trace-free part of $T_{ab}$ will be denoted by $\mathring{T}_{ab}$, i.e.
\begin{equation}
\mathring{T}_{ab } := T_{ab } -\frac{1}{n} g_{ab} T, \quad T:= g^{cd} T_{cd}
\end{equation}
for some $g_{ab} \in \mathbf{c}$, whereas the trace-free symmetrized part of $T_{ab}$ by $T_{\left(ab \right)_0}$, i.e.
\begin{equation}
T_{\left(ab \right)_0} := T_{(ab)} - \frac{1}{n} g_{ab} T.
\end{equation}
The convention that we use for the Riemann tensor associated with a metric $g_{ab} \in \mathbf{c}$ is as follows,
\begin{equation} \label{riemdef}
\left[ \nabla_a,  \nabla_b \right] v^c = R_{ab}{}^c{}_d v^d,
\end{equation}
where $\nabla_a$ is the Levi-Civita connection of $g_{ab}$. The following decomposition holds,
\begin{equation} \label{riemdec}
R_{abcd} = C_{abcd} + 2 \left(g_{c[a}P_{b]d} + g_{d[b}P_{a]c}\right),
\end{equation}
where $C_{abcd}$ is the (fully trace-free) Weyl tensor and
\begin{equation}
P_{ab} :=  \frac{1}{n-2} \left( R_{ab} - \frac{R}{2 \left( n-1 \right)} g_{ab}\right)
\end{equation}
is the Schouten tensor. We will use $J$ to denote its trace, i.e. $J := g^{ab} P_{ab}$. Lastly, we have
\begin{equation} \label{bachcotdef}
\begin{aligned}
A_{abc} & := 2 \nabla_{[b}P_{c]a}, \\
B_{ab} & := - \nabla^c A_{abc} + P^{dc} C_{dacb},
\end{aligned}
\end{equation} 
where $A_{abc}$ and $B_{ab}$ are Cotton and Bach tensors, respectively. It should be noted that the Bianchi identities imply
\begin{equation} \label{bianchi}
 \left( n-3 \right) A_{abc}= \nabla^d C_{dabc}.
\end{equation}

\medskip

When working with a hypersurface (i.e. a codimension one embedded
submanifold) $\Sigma$ with normal vector field $n^a$, we shall
identify $T\Sigma$ with the subbundle $n^\perp$ of $TM|_{\Sigma}$
consisting of tangent vectors orthogonal to $n^a$. Similarly, the
hypersurface cotangent bundle $T^*\Sigma$ will be identified with the
annihilator of $n^a$ in $T^*M|_\Sigma$. In this way we use the same
abstract indices for $T\Sigma$ and $T^*\Sigma$ as we use for,
respectively, $TM$ and $T^*M$. With this understood, the quantities intrinsic
to $\Sigma$ will be denoted by a bar. For example,
given a metric $g_{ab}$ on $M$, $\overline{g}_{ab}$ denotes the metric induced on $\Sigma$ (i.e. the first fundamental form) and $\overline{\nabla}_a$ is its Levi-Civita
connection. Note also that with these conventions $ \overline{g}_a{}^c$ (index raised using $g^{bc}$) gives, by a single contraction, the orthogonal projection from $TM$ to $T\Sigma$.   
We will use superscript $\top$ when working with the
orthogonal projections of the ambient tensors to the hypersurface tensor bundles.
e.g. $T^{\top}_{ab}:= \overline{g}_a{}^c \overline{g}_{b}{}^d T_{cd}$.
The abstract index $n$ will denote contraction with the normal vector
$n^a$, e.g.
\begin{equation}
T^{\top}_{a n b n} =  \overline{g}_a{}^c  \overline{g}_{b}{}^d T_{c f d h} n^f n^h,
\end{equation}
and $\mathring{\top}$ will be used when considering the trace-free part of the projection.

 The second fundamental form $K_{ab}$ of $\Sigma$ is defined in terms of its normal vector $n^a$ and can be decomposed as follows,
\begin{equation}
K_{ab} := \left( \nabla_a n_b \right)^{\top}= \mathring{K}_{ab} +\frac{H}{n-1} \overline{g}_{ab}, \quad H:=  \overline{g}^{cd} K_{cd},
\end{equation}
where $\mathring{K}_{ab}$ denotes the traceless part of $K_{ab}$ and $H$ will be called the mean curvature of $\Sigma$.

Lastly, we will use $\overset{\Sigma}{=}$ when an ambient quantity is evaluated on $\Sigma$, e.g.
\begin{equation}
n_a n^a \overset{\Sigma}{=} -1
\end{equation}
indicates that $n^a$ has a negative unit norm on this hypersurface.
\section{Preliminaries} \label{secprel}
Let $(\widetilde{M}, \widetilde{g}_{ab} )$ be a spacetime that satisfies the Einstein field equations with positive cosmological constant $\Lambda$, 
\begin{equation} \label{efes}
\widetilde{R}_{ab} - \frac{1}{2} \widetilde{R} \widetilde{g}_{ab} + \Lambda \widetilde{g}_{ab} = \widetilde{T}_{ab},
\end{equation}
where $\widetilde{T}_{ab}$ is the stress-energy tensor. Central to this paper is the notion of the asymptotically de Sitter spacetime. We will work with the following definition.
\begin{definition} \label{defadS}
A spacetime $(\widetilde{M}, \widetilde{g}_{ab} )$ is asymptotically de Sitter if there exists a manifold $M$, with boundary $\Sigma$ and a metric $\mathtt{g}_{ab}$, such that
\begin{itemize}
\item there is an embedding $\varphi: \widetilde{M} \to M$ such that $\varphi (\widetilde{M}) =  M \setminus \Sigma$,
\item the metric $\mathtt{g}_{ab} \in \mathbf{c} $ is regular on $M$
  and satisfies $\widetilde{g}_{ab} = \Omega^{-2} \mathtt{g}_{ab}$ (on $\widetilde{M}$) for
  some smooth non-negative function $\Omega$ on $M$,
\item $\Omega$ is a defining function of the boundary $\Sigma$, i.e. $\Sigma =\Omega^{-1}(0)$ and $d \Omega $ is nowhere zero on $\Sigma$,
\item the stress-energy tensor $\widetilde{T}_{ab}$ vanishes along $\Sigma$.
\end{itemize}
\end{definition}
The boundary $\Sigma$ is often called the conformal infinity of $(\widetilde{M}, \widetilde{g}_{ab})$.

\subsection*{Tractor Calculus} 
\hfill

Here by a conformal manifold $(M,\cc)$ we mean a smooth manifold of
dimension $n\geq 3$ equipped with an equivalence class $\cc$ of
metrics, where $g_{ab}$, $\widehat{g}_{ab} \in \cc$ means that $\widehat{g}_{ab}=\Theta^2
g_{ab}$ for some smooth positive function $\Theta$. On a general conformal
manifold $(M,\cc)$, there is no distinguished connection on $TM$. But
there is an invariant and canonical connection on a closely related
bundle, namely the conformal tractor connection on the standard
tractor bundle, see \cite{BEG,CGtams}.

Here we review the basic conformal tractor calculus on
pseudo-Riemannian and conformal manifolds. See
\cite{Go-Pet-lap,Curry-G} for more details.
Unless stated otherwise, every calculation will be done with the use of generic $g_{ab} \in
\cc$. Hence, we will omit the superscript $g$ in the objects
determined by this metric, e.g. $\nabla_a$ will be used instead of
$\nabla^g_a$.

 On any manifold $M$ of dimension $n$ the line bundle
$\mathcal{K}:=(\Lambda^{n} TM)^{\otimes 2}$ of volume densities is
canonically oriented and thus one may take oriented roots of it: Given
$w\in \mathbb{R}$ we set
\begin{equation} \label{cdensities}
\ce[w]:=\mathcal{K}^{\frac{w}{2n}} ,
\end{equation}
and refer to this as the bundle of conformal densities. For any vector bundle $\cV$, we write $\cV[w]$ to mean 
$\cV[w]:=\cV\otimes\ce[w]$.
For example, $\ce_{(ab)}[w]$ denotes the symmetric second tensor power of the cotangent bundle tensored with $\ce[w]$, i.e. $S^2T^* M\otimes \ce[w]$ on $M$. When discussing bundles on $\Sigma$ a $\overline{\ce}$ symbol will be used, e.g. $\overline{\ce}{}^{a}$ is the tangent bundle of this hypersurface.

On a conformal structure there is a canonical section $\bg_{ab}\in
\Gamma(\ce_{(ab)})[2]$. This has the property that for each positive
section $\si\in \Gamma(\ce_+ [1])$ (called a {\em scale})
$g_{ab}:=\si^{-2}\bg_{ab}$ is a metric in $\cc$. Moreover, the Levi-Civita connection of $g_{ab}$ preserves $\si$ and therefore $\bg_{ab}$. Thus it makes sense to use the conformal metric to raise and lower indices, even when we have nominated a metric $g_{ab} \in \cc$ to split the tractor bundles and determine a Levi-Civita connection. It turns out that this simplifies many computations, and so (following the mentioned references) in this section we will do that without further mention. 

Considering Taylor series for sections of $\ce[1]$ one recovers the jet exact sequence at 2-jets,
\begin{equation}\label{J2}
0\to \ce_{(ab)}[1]\stackrel{\iota}{\to} J^2\ce[1]\to J^1\ce[1]\to 0 .
\end{equation}
Then given the conformal structure we have the orthogonal decomposition in trace-free and trace parts 
\begin{equation}
\ce_{ab}[1]= \ce_{(ab)_0}[1]\oplus \bg_{ab}\cdot \ce[-1] . 
\end{equation}
Thus we can canonically quotient $J^2\ce[1]$ by the image of
$\ce_{(ab)_0}[1]$ under $\iota$ (in (\ref{J2})) to form the bundle
$\cT^*$, called the conformal cotractor bundle.

Given a choice of metric $g_{ab} \in\cc$, the formula
\begin{equation}\label{thomas_D}
  \sigma \mapsto \frac{1}{n} {[D_A \sigma]}_g := 
  \begin{pmatrix}
    \sigma \\ 
    \nabla_a \sigma \\ 
    - \frac{1}{n} \left( \Delta + J \right) \sigma 
  \end{pmatrix}
\end{equation}
(where $\Delta $ is the Laplacian $\nabla^a\nabla_a$)
gives a second-order differential operator on $\ce[1]$ which is a linear map $J^2 \ce[1] \to \ce[1] \oplus \ce_a [1] \oplus \ce[-1]$ that clearly factors through $\cT^*$ and so determines an isomorphism
\begin{equation}\label{T_isom}
  \cT^* \stackrel{\sim}{\longrightarrow} {[\cT^*]}_g = \ce[1] \oplus \ce_a [1] \oplus \ce[-1].
\end{equation}
The tractor defined in (\ref{thomas_D}) will be called the scale tractor corresponding to the scale $\sigma$ and denoted by $I_{\sigma}$, i.e.
\begin{equation}
I_{\sigma } := \frac{1}{n} D_A \sigma.
\end{equation}

In subsequent discussions, we will use~\eqref{T_isom} to split the
tractor bundles without further comment.  Thus, given $g_{ab} \in \cc$, an
element $V^A$ of $\ce^A$ may be represented by a triple
$(\si,\mu_a,\rho)$, or equivalently by
\begin{equation}\label{Vsplit}
  V_A=\si Y_A+\mu_a Z_A{}^a+\rho X_A.
\end{equation}
The last display defines the algebraic splitting operators
$Y:\ce[1]\to \cT^*$ and $Z :T^*M[1]\to \cT^*$ (determined by the
choice $g_{ab} \in \cc$) which may be viewed as sections $Y_A\in
\Gamma(\ce_A[-1])$ and $Z_A{}^a\in \Gamma(\ce_A{}^a [-1])$.  We call
these sections $X_A, Y_A$ and $Z_A{}^a$ \emph{tractor projectors}.

It is straightforward to verify that the equation
\begin{equation}\label{AE}
\nabla_{(a}\nabla_{b)_0} \si+ P_{(ab)_0}\si =0
\end{equation}
on conformal densities $\si\in \Gamma(\ce[1])$ is conformally
invariant. As it is overdetermined, solutions may not exist and indeed 
it is well known that positive ones are equivalent to
vacuum-Einstein metrics in the conformal class: $\si \in
\Gamma(\ce_+[1])$ solves (\ref{AE}) is equivalent to
$P^{\hat{g}}_{(ab)_0}=0$, where $\hat{g}_{ab} =\si^{-2}\bg_{ab}$. More
generally, non-trivial solutions are non-vanishing on an open dense set, on
which they determine a vacuum-Einstein metric
\cite{G-almost,Curry-G}. Thus (\ref{AE}) is sometimes called the {\em almost Einstein equation}.

Given a metric $g_{ab} \in \cc$, the tractor connection is given by the formula
 \begin{equation}\label{tr-conn}
 \nabla_a^\cT  \begin{pmatrix}
    \sigma \\ 
    \mu_b \\ 
    \rho
  \end{pmatrix} :=\begin{pmatrix}
   \nabla_a \sigma-\mu_a \\ 
    \nabla_a\mu_b+P_{ab}\si+\bg_{ab}\rho \\ 
    \nabla_a\rho- P_{ac}\mu^c
  \end{pmatrix}
\end{equation}
and the equation of parallel transport is equivalent to a
prolongation of the almost Einstein equation (\ref{AE}), see
\cite{BEG,Curry-G} (and Section \ref{secaem} below). Thus, in particular, solutions
$V_A$ of $\nabla_a^\cT V_A=0$ are equivalent to solutions $\si$ of
(\ref{AE}) with the explicit relations:
$$
\si=X^AV_A \qquad\mbox{and}\qquad V =I_{\si} .
$$
The tractor bundle is also equipped with a tractor metric $h_{AB} \in \Gamma \left(\ce_{(AB)} \right)$, defined using the mapping
\begin{equation}
[V_A]_g = \begin{pmatrix}
    \sigma \\ 
    \mu_a \\ 
    \rho
  \end{pmatrix} \to \mu_a \mu^a +2  \si \rho =: h \left(V,V \right) 
\end{equation}
combined with the polarization identity. It can be checked that the tractor metric is preserved by $\nabla_a^\cT$, i.e. $ \nabla_a^\cT h_{AB} =0$.

The curvature of the tractor connection $\kappa_{abCD}$ can be recovered with the use of the following relation,
\begin{equation} \label{trc1}
2 \nabla_{[a}^\cT \nabla_{b]}^\cT V^C = \kappa_{ab}{}^{C}{}_D V^D \quad \mathrm{for \ all} \quad V^C \in \Gamma \left(\ce^A \right),
\end{equation}
and can be written in terms of tractor projectors as
\begin{equation} \label{trc2}
\kappa_{abCD} = A_{cab} Z_{C}{}^c X_D - A_{cab} X_C Z_D{}^c + C_{abcd} Z_{C}{}^c Z_D{}^d.
\end{equation}

\subsection*{The scale singularity set and the normal tractor}
\hfill

Since metrics $g_{ab} \in \cc$ correspond to section $\si\in
\Gamma(\ce[1])$ via $g_{ab} =\si^{-2}\bg_{ab}$, points where $\si$
vanishes are conformal singularities of $g_{ab}$. It is elementary to
verify that if the ``length squared'' of the scale tractor, meaning
$I_{\si}^2:= h \left( I_{\si}, I_{\si} \right)$, is nowhere zero then
the zero locus $\Sigma = \mathcal{Z}(\si)$ of $\si$ is a smoothly embedded
separating hypersurface.

Alternatively, if we are considering a manifold with boundary, it can be the case that $\Sigma =\mathcal{Z}(\si)$ is the
boundary. Along such $\mathcal{Z} (\si)$ there is a conformally
invariant tractor analogue of the normal vector called the normal tractor \cite{BEG} -- a section $N_A$ of $\cT|_\Sigma$ that is given in a metric $g_{ab} \in \cc$ by the formula
\begin{equation} \label{normaltr}
N_A \stackrel{g}{=}\left(\begin{array}{c}0\\
n_a\\
-\frac{H}{n-1}\end{array}\right),
\end{equation}
where $n_a \in \Gamma \left( \mathcal{E}_a[1] \right)$ and $H \in \Gamma \left( \mathcal{E}[-1] \right)$ are the densities corresponding to the normal vector and the mean curvature of  $\mathcal{Z} (\si)$, which will be defined in Section \ref{secaem} for asymptotically de Sitter spacetimes.

If $I_{\si}^2 = \pm 1$ along $\Sigma$, then a rather nice feature is captured by the following proposition \cite{G-almost,Curry-G}:
\begin{proposition}\label{ascN}
Let $(M,\cc)$ be a pseudo-Riemannian conformal structure and suppose that   
a scale tractor $I_{\si}$ has a scale singularity set $\Sigma=\mathcal{Z}(\si)\neq \emptyset$ and
$I_{\si}^2=\pm 1 + \si^2 f$ for some smooth (weight $-2$) density $f$. Then we
have $N=I_{\si}|_\Sigma$, where $N_A$ is the normal tractor. 
\end{proposition}
In particular, this holds if the norm of the scale tractor is equal to a constant on the boundary $\Sigma$ of
asymptotically de Sitter spacetimes. This will be explored in Section
\ref{secaem}.

Along a hypersurface (or boundary) $\Sigma$, the normal tractor can
be used to introduce the tractor projection operator,
\begin{equation}
\Pi_A{}^B = \delta_A{}^B \mp N_A N^B,
\end{equation}
with signs reflecting the cases $N_A N^A = \pm 1$. The image of
$\Pi_A{}^B$ is isomorphic to the hypersurface tractor bundle
$\overline{\cT}$, where the isomorphism is given by \cite{G-almost,Curry-G}
\begin{equation} \label{hypspl}
\begin{split}
\overline{X}^A & = X^A, \\
\overline{Y}^A & = Y^A  \mp Z^{Aa} \frac{H}{n-1} n_a  \pm  \frac{H^2}{2 \left(n-1 \right)^2} X^A, \\
\overline{Z}^{Aa} & =  \overline{g}^a{}_b Z^{A b}.
\end{split}
\end{equation}
\subsection*{Elements of tractor calculus} 
\hfill

We will use the symbol $\ce^{\Phi}[w]$ to denote any tractor bundle of weight $w$.
The operator $D_A$ in (\ref{thomas_D}) generalizes to a conformally
invariant differential operator on sections of $\ce^\Phi[w]$ \cite{BEG,Go-Pet-lap}
\begin{equation}
D_A: \ce^\Phi[w]\to \ce_A\otimes \ce^\Phi[w-1],
\end{equation}
and is given by 
\begin{equation}\label{tD}
V \mapsto [D_A V]_g:= \left(\begin{array}{c}
(n+2w-2)w V \\
(n+2w-2) \nabla_a V \\
- (\Delta V + w J V)
\end{array}\right)
  \end{equation}
in a metric $g_{ab} \in \cc$. On the right-hand side of the display
the $\nabla_a$ is the Levi-Civita connection coupled with the tractor
connection (including in the Laplacian).  This is usually called the
tractor-D, or {\em Thomas-D}, operator. It is often useful to use its
rescaled version, denoted here by $\widehat{D}_A$, i.e.
\begin{equation}\label{tDhat}
V \mapsto [\widehat{D}_A V]_g:= \left(\begin{array}{c}
w V \\
 \nabla_a V \\
- \frac{1}{(n+2w-2)} (\Delta V + w J V)
\end{array}\right),
  \end{equation}
which is defined for $w \neq 1- \frac{n}{2}$.  The Thomas-D operator may be combined with the scale tractor to produce a canonical degenerate Laplacian type differential operator \cite{G-DN}, namely 
\begin{equation} \label{deglap}
\ID:= I_{\si}^AD_A.
\end{equation}
This acts on any weighted tractor bundle, preserving its tensor type
but lowering the weight,
\begin{equation}
\ID : \ce^\Phi[w]\to  \ce^\Phi[w-1].
\end{equation}
It can be expanded in terms of a metric $g_{ab} \in \cc$ to yield 
\begin{equation}\label{degLa}
\ID \stackrel{g}{=} -\si \Delta
+(n+2w-2)[(\nabla^a \si)\nabla_a - \frac{w}{n} (\Delta \si)] -\frac{2w}{n}(n+w-1)\si J ~
\end{equation} on $\ce^\Phi [w]$.
 Now if we calculate in the metric $g_{ab}=\si^{-2}\bg_{ab}$, away from  the
 zero locus of $\si$, and trivialize the densities accordingly,  then $\si$ is represented by 1 in the trivialization, and we have
\begin{equation}
 \ID \stackrel{g}{=}-  \Big(\Delta + \frac{2w(n+w-1)}{n} J \Big).
\end{equation}
On the other hand, looking again to~(\ref{degLa}), we see that $\ID$
degenerates along the conformal infinity $\Sigma = \mathcal{Z}(\si)$
(assumed non-empty), and there the operator is first order. In
particular, if the structure is asymptotically almost scalar constant
in the sense that $I_{\si}^2=\pm 1+ \si^2 f$ for some smooth (weight
$-2$) density $f$, then along $\Sigma$
\begin{equation}
\ID \overset{\Sigma}{=} (n+2w-2) \delta_R,
\end{equation}
where $\delta_R$  is the conformal Robin operator,
\begin{equation} \label{delrob}
\delta_R \stackrel{g}{=}  \nabla_n - w \frac{H}{n-1} ,
\end{equation}
of~\cite{Cherrier,BrGo-non} (twisted with the tractor connection).

Given a $t_{ab} \in \Gamma \left( \ce_{(ab)_0} [w]\right)$ with $w \neq 2-n, \ 3-n$ we can define a map $p$ which inserts $t_{ab}$ into a symmetric trace-free tractor $T_{AB}$,
\begin{equation}
p: \Gamma \left( \ce_{\left(ab\right)_0} \left[w \right]\right) \to \Gamma  \left( \ce_{\left(AB\right)_0} \left[w-2\right]\right), \quad w \neq 2-n, \ 3-n,
\end{equation}
which is given by
\begin{equation}
p \left( t_{ab} \right) := Z_A{}^a Z_B{}^b t_{ab} -  \frac{ 2 \nabla \cdot t_c }{n+w-2} X_{(A} Z_{B)}{}^c + \frac{\nabla \cdot \nabla \cdot t + \left( n+w-2 \right) P \cdot t}{\left(n+w -2 \right) \left(n+w-3 \right)} X_A X_B.
\end{equation}
We have
\begin{equation}
D^A T_{AB} = 0 = X^A T_{AB}.
\end{equation}
It can be seen that a general property of a symmetric tractor $T_{AB....} \in \Gamma \left( \ce_{\left(AB...\right)} \left[w\right] \right)$ with $X^A T_{AB...} = 0$ is that the tensor $Z^A{}_a Z^B{}_b ... T_{AB...}$ is conformally invariant. Therefore, the notion of extraction operator $p^*$ can be defined for such tractors, i.e. let
\begin{equation}
p^*: \Gamma \left( \ce^{\mathring{X}}_{\left(AB...\right)} \left[w\right] \right) \to \Gamma \left( \ce_{\left(ab...\right)} \left[w+v \right] \right),
\end{equation}
where $\ce^{\mathring{X}}$ denotes the tractor bundle whose sections vanish when contracted with $X^A$ and $v$ is the valence of $T_{AB...}$. In particular, 
\begin{equation}
p^* \left( T_{AB...} \right) := Z^A{}_a Z^B{}_b \  ... \  T_{AB...}, \quad T_{AB...} \in \Gamma \left( \ce^{\mathring{X}}_{\left(AB...\right)} \right).
\end{equation}
The crucial property of the operators $p$ and $p^*$ is that
\begin{equation}
p^* \circ p = \mathrm{Id},
\end{equation}
i.e. the insertion operator is the right-inverse of the extraction operator. (Such operators are often called differential splitting operators.)

In order to obtain a section of the $\ce^{\mathring{X}}$ bundle, from section of a generic tractor bundle, a tractor {\em projection} differential operator $r$ can be used.

We will restrict ourselves to traceless tractors of valence 2, i.e.
\begin{equation}
\begin{split}
r \left( T^{AB} \right)   := & 
T^{AB}-
\frac{2}w \widehat{D}^{(A} \big(X_C T^{|C|B)_0 }\big)
+\tfrac1{w(w+1)} \widehat{D}^{(A} \widehat{D}^{B)_0} \big(X_C X_D T^{CD}\big)
\\
& -\tfrac{8}{wn(n+2w+2)}
X^{(A} \widehat{D}^{B)_0} \big(\widehat{D}_C (X_D T^{CD})\big),
\end{split}
\end{equation}
for $w \neq 0, -1, - \frac{n}{2}, -1 - \frac{n}{2}, -2 - \frac{n}{2}$ and $T_{AB} \in \Gamma \left( \ce_{\left(AB\right)_0} \left[w\right] \right)$.

\begin{remark}
If $n \geq4$ then the injection, extraction and projection 
operators have their hypersurface analogues
$\overline{p}$, $\overline{p}^*$, and $\overline{r}$ which can be
defined with the use of hypersurface tractor connection, the
corresponding hypersurface tractor operators, and hypersurface tractor
projectors (\ref{hypspl}). The notion of hypersurface operators can
still be introduced in 3 dimensions with an additional M{\"o}bius
structures over the conformal boundary.
\end{remark}

\section{The almost-Einstein-matter equation and its consequences} \label{aes} \label{secaem}
Let $(\widetilde{M}, \widetilde{g}_{ab})$ be an asymptotically de Sitter spacetime. After taking a divergence of the Einstein field equations (\ref{efes}) one arrives at the matter continuity equation,
\begin{equation} \label{divt}
\widetilde{\nabla}^a \widetilde{T}_{ab} = 0.
\end{equation}
 We can utilize a decomposition of the Ricci tensor $\widetilde{R}_{ab}$ into the Schouten tensor $\widetilde{P}_{ab}$ and its trace $\widetilde{J}$,
\begin{equation}
\widetilde{R}_{ab} = (n-2) \widetilde{P}_{ab} + \widetilde{J} \widetilde{g}_{ab},
\end{equation}
to obtain a trace-free part of the Einstein field equations (\ref{efes}),
\begin{equation} \label{efes2}
  \widetilde{P}_{ab} - \frac{1}{n} \widetilde{J} \widetilde{g}_{ab} = \frac{1}{n-2} \left( \widetilde{T}_{ab} - \frac{1}{n} \widetilde{T} \widetilde{g}_{ab} \right).
\end{equation}
Suppose that
 \begin{equation}
 \widetilde{g}_{ab} = f^2 g_{ab},
 \end{equation}
 where $f$ is a positive smooth function, i.e.\ $g_{ab}$ is in the conformal class of $\widetilde{g}_{ab}$. Then the transformation rule for the Schouten tensor reads
\begin{equation} \label{schout_tr}
\widetilde{P}_{ab} = P_{ab} - \nabla_a \Upsilon_b + \Upsilon_a \Upsilon_b - \frac{1}{2} g_{ab} \Upsilon^2, 
\end{equation}
where $\Upsilon_a = \nabla_a \log f$. We will now use this to relate (\ref{efes2}) with the almost Einstein equation (\ref{AE}).
Let $\si$ and $\widetilde{\si}$ be scales corresponding to $g_{ab}$ and $\widetilde{g}_{ab}$ respectively, i.e.
\begin{equation}
\widetilde{\sigma}^2 \widetilde{g}_{ab} = \mathbf{g}_{ab} = \sigma^2 g_{ab},
\end{equation}
where $\textbf{g}_{ab}$ is the conformal metric. When working in the scale $\sigma$, the density $\widetilde{\sigma}$ is determined by $f^{-1}$ and 
\begin{equation}
\nabla_a \nabla_b \widetilde{\sigma} = \widetilde{\sigma} \left(- \nabla_a \Upsilon_b + \Upsilon_a \Upsilon_b \right).
\end{equation}
The Einstein field equations (\ref{efes2}) can now be written as 
\begin{equation} \label{efes3}
\nabla_a \nabla_b \widetilde{\sigma} + \widetilde{\sigma} P_{ab} - \frac{1}{n} \mathbf{g}_{ab} \left( \Delta \widetilde{\sigma} + \widetilde{\sigma} J \right) = \frac{\widetilde{\sigma}}{n-2} \mathring{\widetilde{T}}_{ab},
\end{equation}
where
\begin{equation}
 \mathring{\widetilde{T}}_{ab} := \widetilde{T}_{ab} - \frac{1}{n} \widetilde{g}_{ab} \widetilde{T}
\end{equation}
is the traceless part of the physical stress-energy tensor.
\begin{definition}
The scale $\widetilde{\si} $ will be called the almost-Einstein-matter scale.
\end{definition}
Let $\tau_{ab} \in \mathcal{E}_{ab} [-q]$ be the conformally weighted
tensor, corresponding to $\widetilde{T}_{ab}$, defined by
\begin{equation}
\tau_{ab} := \widetilde{\sigma}^{-q} \widetilde{T}_{ab},
\end{equation}
where we assume that $q \in \{0,1,2 \}$. In order to attribute a
physical meaning to the parameter $q$ consider a scale $\eta$
corresponding to the regular metric $\mathtt{g}_{ab}$ which defines
the conformal extension of the spacetime $( \widetilde{M},
\widetilde{g}_{ab} )$ (cf. Definition \ref{defadS}). When
working in the scale $\eta$, the almost-Einstein-matter scale
$\widetilde{\sigma}$ is characterized by the defining function of the
conformal boundary $\Omega$, i.e. $\widetilde{\sigma} = \Omega \eta$,
so
\begin{equation} \label{tphysunphys}
\widetilde{T}_{ab} = \Omega^q \mathtt{T}_{ab},
\end{equation}
where $\mathtt{T}_{ab}$ is the {\em unphysical stress-energy
  tensor} ($\mathtt{T}_{ab}$ and $\Omega$ are, respectively,  $\tau_{ab}$ and $\tilde{\sigma}$ in the scale $\eta$).
Hence, if $\mathtt{T}_{ab}$ is regular
everywhere, then $q$ characterizes the decay of the stress-energy
tensor $\widetilde{T}_{ab}$ when one approaches conformal infinity
of the asymptotically de Sitter spacetime. Based on this observation
the case $q=0$ should be excluded from the analysis, as such matter
fields do not vanish on the conformal boundary. Nevertheless, it fits naturally in the tractor calculus approach presented here, so we
include it in the computations.

The traceless part of $\widetilde{T}_{ab}$ can be related to the traceless part of $\tau_{ab}$ in the following way,
\begin{equation}
\mathring{\widetilde{T}}_{ab} =  \widetilde{\sigma}^q \left(  \tau_{ab} - \frac{1}{n} \mathbf{g}_{ab} \tau \right) =  \widetilde{\sigma}^q \mathring{\tau}_{ab},
\end{equation}
where $\tau = \mathbf{g}^{cd} \tau_{cd} \in \mathcal{E} [-q -2]$ is a density corresponding to the trace $\widetilde{T}$. Equation (\ref{efes3}) now reads
\begin{equation} \label{aleinm}
\nabla_a \nabla_b \widetilde{\sigma} + \widetilde{\sigma} P_{ab} - \frac{1}{n} \mathbf{g}_{ab} \left( \Delta \widetilde{\sigma} + \widetilde{\sigma} J \right) = \frac{\widetilde{\sigma}^{q+1}}{n-2} \mathring{\tau}_{ab},
\end{equation}
and will be called the almost-Einstein-matter equation.

\subsection*{Prolongation of the almost-Einstein-matter equation} 
\hfill

Let
\begin{equation}
\rho := - \frac{1}{n} \left( \Delta \widetilde{\sigma} + \widetilde{\sigma} J \right).
\end{equation}
By taking two different contractions of the covariant derivative of (\ref{aleinm}) we obtain
\begin{equation}
\begin{split}
& \Delta \nabla_a \widetilde{\sigma} + \widetilde{\sigma} \nabla_c P_a{}^c + P_a{}^b \nabla_b \widetilde{\sigma} + \nabla_a \rho = \frac{q+1}{n-2} \widetilde{\sigma}^q \mathring{\tau}_a{}^b \nabla_b \widetilde{\sigma}  + \frac{\widetilde{\sigma}^{q+1}}{n-2} \nabla_b \mathring{\tau}_a{}^b, \\
& \nabla_a \Delta \widetilde{\sigma}  + \widetilde{\sigma} \nabla_a J + J \nabla_a \widetilde{\sigma} + n \nabla_a \rho = 0.
\end{split}
\end{equation}
After taking a difference, using the contracted Bianchi identity,
\begin{equation}
\nabla_b P_a{}^b = \nabla_a J,
\end{equation}
and expressing the commutator $[\nabla_c, \Delta]$ in terms of the Ricci tensor one gets
\begin{equation} \label{prol1}
\nabla_a \rho - P_a{}^b  \nabla_b \widetilde{\sigma} = - \frac{\widetilde{\sigma}^q}{\left(n-1 \right) \left(n-2 \right)}\left(  \left(q+1\right)  \mathring{\tau}_a{}^b \nabla_b \widetilde{\sigma}  + \widetilde{\sigma} \nabla_b \mathring{\tau}_a{}^b \right).
\end{equation}
Hence, the second-order almost-Einstein-matter equation (\ref{aleinm}) is equivalent to the first-order system of three equations,
\begin{equation} \label{prol2}
\begin{split}
& \nabla_a \widetilde{\sigma} - \mu_a =0, \\ & \nabla_a \mu_b +
  \widetilde{\sigma} P_{ab} + \rho \mathbf{g}_{ab} =
  \frac{\widetilde{\sigma}^{q+1}}{n-2} \mathring{\tau}_{ab}, \\
  &
  \nabla_a \rho - P_a{}^b \mu_b = -
  \frac{\widetilde{\sigma}^q}{\left(n-1 \right) \left(n-2
    \right)}\left( \left(q+1\right) \mathring{\tau}_a{}^b \mu_b +
  \widetilde{\sigma} \nabla_b \mathring{\tau}_a{}^b \right),
\end{split}
\end{equation}
in three variables $\widetilde{\sigma}$, $\mu_a$ and $\rho$ (compare
with the definition (\ref{tr-conn}) of the tractor
connection). Due to the presence of matter fields, the r.h.s. of
the above is non-zero, which has direct consequences for the
derivative of the almost-Einstein-matter scale tractor
$I_{\widetilde{\si}}$. Unlike in the standard picture presented in
Section \ref{secprel}, $I_{\widetilde{\si}}$ will no longer be
parallel with respect to the tractor connection. This fact is explored
in more detail below.

\subsection*{The almost-Einstein-matter scale tractor} 
\hfill

Let
\begin{equation}
I_{\widetilde{\sigma}}:= \left(\begin{array}{c}
\widetilde{\sigma} \\
 \nabla_b \widetilde{\sigma} \\
-  \frac{1}{n}(\Delta \widetilde{\sigma} + J \widetilde{\sigma})
\end{array}\right) .
\end{equation}
This is 
  the $\widetilde{\sigma}$-scale tractor. In the scale $\widetilde{\sigma}$
 \begin{equation} \label{i2}
I_{\widetilde{\sigma}}^2 = - \frac{\widetilde{R}}{n \left( n-1 \right)}  = -  \lambda+ \frac{2 \widetilde{\sigma}^{q+2} \tau }{n \left( n-1 \right) \left( n-2 \right) },
 \end{equation}
 where 
 \begin{equation}
\lambda :=  \frac{2  \Lambda }{ \left( n-1 \right) \left( n-2 \right) } .
 \end{equation}
If the matter fields are present, then the scale tractor will no longer be parallel, i.e.
\begin{equation} \label{deri}
\nabla_a I_{\widetilde{\sigma}}  = \frac{\widetilde{\sigma}^{q}}{ \left( n-1 \right) \left( n-2 \right)}   \left(\begin{array}{c}
0 \\
\widetilde{\sigma} \left( n-1 \right) \mathring{\tau}_{ab} \\
- \left(q+1\right)  \mathring{\tau}_a{}^b \nabla_b \widetilde{\sigma}  - \widetilde{\sigma} \nabla_b \mathring{\tau}_a{}^b
\end{array}\right),
\end{equation}
which is a consequence of the almost-Einstein-matter equation (\ref{aleinm}) and its prolongation (\ref{prol2}). Calculating the conformal transformation
\begin{equation}\label{c1}
 \nabla^{\widetilde{g} }_b \mathring{\tau}_a{}^b =  \nabla^g_b \mathring{\tau}_a{}^b + \left( n-2 - q  \right) \Upsilon^b  \mathring{\tau}_{ab},
 \end{equation}
also verifies that the r.h.s.\ of (\ref{deri}) transforms as a
tractor. Moreover, a direct calculation yields
\begin{equation} \label{checki}
\nabla^{\widetilde{g} }_a I_{\widetilde{\sigma}}^2 = 2
I_{\widetilde{\sigma}} \cdot \nabla^{\widetilde{g} }_a
I_{\widetilde{\sigma}} = - \frac{2 \widetilde{\sigma}^{q+2}}{ \left(
  n-1 \right) \left( n-2 \right)}  \nabla^{\widetilde{g} }_b
\mathring{\tau}_a{}^b,
\end{equation}
where (\ref{deri}) and (\ref{c1})  have been used. The continuity equation (\ref{divt}) implies
\begin{equation}
0 = \widetilde{\sigma}^{q+2} \left( \nabla^{\widetilde{g} }_b \mathring{\tau}_a{}^b + \frac{1}{n} \nabla^{\widetilde{g}}_a \tau \right),
\end{equation}
so outside of a zero locus of $\widetilde{\sigma}$ equation (\ref{checki}) agrees with the derivative of (\ref{i2}).

\subsection*{Trace-free second fundamental form $\mathring{K}_{ab}$ of $\Sigma$}
\hfill

The conformal infinity $\Sigma$ of the asymptotically de Sitter
spacetime is an embedded hypersurface of its conformal extension $(M,
\mathtt{g}_{ab})$. As the almost-Einstein-matter scale is a defining
density of $\Sigma$, there is a natural notion of a normal vector
associated with it. It can be used to construct the first two conformal fundamental forms of $\Sigma$.

Let
\begin{equation}
n_a := \frac{1}{\sqrt{\lambda}} \nabla_a \widetilde{\sigma}. 
\end{equation}
It can be seen that $n_a$ is a conformal density of weight $1$, i.e. $n_a \in \Gamma (\mathcal{E}_a[1])$. According to (\ref{i2}), the norm of this vector is as follows,
\begin{equation} \label{nnorm}
n_a n^a = -1 + \frac{2}{ n \lambda }  \widetilde{\sigma} \left( \Delta + J  \right) \widetilde{\sigma} + \mathcal{O} \left( \widetilde{\sigma}^{q+2}\right).
\end{equation}
The first conformal fundamental form of $\Sigma$, its induced conformal metric $ \overline{\mathbf{g}}_{ab}$, can now be defined as
\begin{equation}
  \overline{\mathbf{g}}_{ab}
  := \mathbf{g}_{ab} + \frac{n_a n_b}{1 + \frac{2}{\lambda} \rho \widetilde{\sigma}  }+ \mathcal{O} \left( \widetilde{\sigma}^{q+2}\right)
  \overset{\Sigma}{=}  \mathbf{g}_{ab} + n_a n_b.
\end{equation}

The extrinsic curvature of the conformal boundary (i.e. the second fundamental form) is given by
\begin{equation}
K_{ab} := \overline{\mathbf{g}}_a{}^c \nabla_c n_b |_{\Sigma},
\end{equation}
i.e.
\begin{equation}
\nabla_a n_b\overset{\Sigma}{=} K_{ab} - n_a \nabla_n n_b \bigg|_{\Sigma}.
\end{equation}
However, since 
\begin{equation}
\nabla_{[a} n_{b]} = \frac{1}{\sqrt{\lambda}} \nabla_{[a} \nabla_{b]} \widetilde{\sigma} = 0,
\end{equation}
the $\nabla_n n_a$ can be expressed as
\begin{equation} \label{accel}
\begin{split}
\nabla_n n_a  & = \frac{1}{2} \nabla_a \left( n^b n_b  \right) = -  \frac{1}{\sqrt{\lambda}} n_a   \rho  +   \mathcal{O} \left(  \widetilde{\sigma} \right) =\frac{1}{ n \sqrt{\lambda}} n_a  \Delta \widetilde{\sigma} +   \mathcal{O} \left(  \widetilde{\sigma} \right).
\end{split}
\end{equation}
We can use (\ref{accel}) to compute the mean curvature $H$ of $\Sigma$,
\begin{equation} \label{meancurvdens}
\begin{split}
H & := \overline{\mathbf{g}}^{ab} K_{ab} \overset{\Sigma}{=} \overline{\mathbf{g}}^{ab} \nabla_a n_b \\
&  = \left( \frac{1}{\sqrt{\lambda}} \Delta \widetilde{\sigma} + \frac{1}{2} \nabla_n \left( n_a n^a \right) \right) \bigg|_{\Sigma} =  \frac{n-1}{n \sqrt{\lambda} }  \Delta \widetilde{\sigma} \bigg|_{\Sigma}.
\end{split}
\end{equation}
As a result,
\begin{equation}\label{nablan}
\nabla_a n_b \overset{\Sigma}{=} K_{ab} - \frac{H}{n-1} n_a n_b  = \mathring{K}_{ab} + \frac{H}{n-1} \mathbf{g}_{ab}
\end{equation}
where the decomposition of $K_{ab}$ into its traceless part and the mean curvature $H$ has been used. Thus, we obtain 
\begin{equation}
I_{\widetilde{\sigma}} \overset{\Sigma}{=}  \left(\begin{array}{c}
0 \\
 \sqrt{\lambda} n_a \\
-  \frac{\sqrt{\lambda} }{n-1 }H 
\end{array} \right),
\end{equation}
which is the asymptotically de Sitter analogue of the normal tractor (\ref{normaltr}).
The derivative of $I_{\widetilde{\sigma}}$ evaluated on $\Sigma$ reads
\begin{equation} \label{dern1}
\nabla_a I_{\widetilde{\sigma}}  \overset{\Sigma}{=}  \left(\begin{array}{c}
0  \\
 \sqrt{\lambda} \left(  \nabla_a n_b - \mathbf{g}_{ab} \frac{H}{n-1} \right)  \\
 - \sqrt{\lambda} \left( \frac{1}{n-1} \overline{\nabla}_a H + P^{\top}_{an} \right) - n_a \left( \nabla_n \rho - \sqrt{\lambda} P_{nn} \right)
\end{array} \right),
\end{equation}
where
\begin{equation} \label{dernrhomat}
\nabla_n \rho \overset{\Sigma}{=} \sqrt{\lambda} \left( P_{nn} + \frac{1}{2} \mathring{K}_{ab} \mathring{K}^{ab} + \frac{1}{48} \left(q+1 \right) \left(q+2\right) \widetilde{\sigma}^q \tau \right) \bigg|_{\Sigma}.
\end{equation}
(see \cite[Lemma 3.8]{GW} for a derivation of this identity without the matter fields). The hypersurface Codazzi equation
\begin{equation} \label{hypcod}
\overline{\nabla}_{a}K_{bc} - \overline{\nabla}_b K_{ac} =  R^{\top}_{abc n}
\end{equation}
can be used to simplify (\ref{dern1}). If we use a decomposition (\ref{riemdec}) of the Riemann tensor, then (\ref{hypcod}) reads
\begin{equation} \label{codazzi}
\begin{split}
& \frac{1}{n-1}\overline{\nabla}_a H = \frac{1}{n-2} \overline{\nabla}_b \mathring{K}_a{}^b - P^{\top}_{a n}, \\
&  \overline{\nabla}_{[a}  \mathring{K}_{b]c} - \frac{1}{n-2} \overline{\mathbf{g}}_{c[a}  \overline{\nabla}{}^d \mathring{K}_{b]d} = - \frac{1}{2} C^{\top}_{n c a b}.
\end{split}
\end{equation}
Ultimately, it can be computed that the projected part of (\ref{dern1}) has the following form
\begin{equation} \label{dern2}
\overline{\mathbf{g}}_a{}^b \nabla_b I_{\widetilde{\sigma}}  \overset{\Sigma}{=}  \sqrt{\lambda}  \left(\begin{array}{c}
0  \\
 \mathring{K}_{ab}    \\
 -   \frac{1}{n-2} \overline{\nabla}_b \mathring{K}_c{}^b
\end{array} \right).
\end{equation}
It follows that $ \mathring{K}_{ab} = p^* (\overline{\mathbf{g}}_a{}^b \nabla_b I_{\widetilde{\sigma}} )$ is conformally invariant, as is of course well known.

The normal component of $\nabla_a I_{\widetilde{\sigma}} $ reads
\begin{equation}
\nabla_n I_{\widetilde{\sigma}}  =  \left(\begin{array}{c}
0  \\
 \sqrt{\lambda}  \nabla_n n_b + n_b \rho + P_{n b } \widetilde{\sigma}   \\
\nabla_n \rho - \sqrt{\lambda} P_{n n}
\end{array} \right)   \overset{\Sigma}{=}  \sqrt{\lambda}   \left(\begin{array}{c}
0  \\
0  \\
 \frac{1}{2} \mathring{K}_{ab}  \mathring{K}^{ab} + \frac{\left(q+1 \right) \left(q+2 \right)}{48} \widetilde{\sigma}^q \tau \bigg|_{\Sigma}
\end{array} \right).
\end{equation}

\section{The conformal fundamental forms of the conformal boundary $\Sigma$} \label{fundformsec}
We will move now to the discussion of four-dimensional asymptotically
de Sitter spacetimes and derive constraints relating the conformal
fundamental forms of its conformal infinity with the matter
fields. Before doing so, we will discuss the more fundamental
constraints which appear when the derivative of the almost
Einstein-matter scale tractor is considered.

\subsection*{Constraints on the matter fields on $\Sigma$}
\hfill

The almost-Einstein-matter scale tractor $I_{\widetilde{\sigma}}$ can be used to derive the constraints on the matter fields on the conformal boundary $\Sigma$. From (\ref{deri}) we have
\begin{equation} \label{ideri}
I_{\widetilde{\sigma}} \cdot \nabla_a I_{\widetilde{\sigma}} = \frac{\widetilde{\sigma}^{q+1}}{6} \left( \left(2-q \right) \mathring{\tau}_a{}^b \nabla_b \widetilde{\sigma} - \widetilde{\sigma} \nabla_b \mathring{\tau}_a{}^b\right).
\end{equation}
On the other hand,
\begin{equation} \label{ideri2}
I_{\widetilde{\sigma}} \cdot \nabla_a I_{\widetilde{\sigma}} = \frac{1}{2} \nabla_a I^2_{\widetilde{\sigma}} = \frac{1}{24} \nabla_a \left( \widetilde{\sigma}^{q+2} \tau \right).
\end{equation}
Therefore, from (\ref{ideri}) and (\ref{ideri2}),
\begin{equation} \label{mattersig}
\sqrt{\lambda} \left[ \left(q-2 \right) \tau_{an} +  n_a \tau \right] = - \widetilde{\sigma} \nabla_b \tau_a{}^b.
\end{equation}
After evaluating (\ref{mattersig}) at $\Sigma$ we get
\begin{equation} \label{mattc1}
\left(q-2 \right) \tau_{an} +  n_a \tau \overset{\Sigma}{=} 0,
\end{equation} 
i.e. 
\begin{equation}
\left(q-2 \right) \tau^{\top}_{an}  \overset{\Sigma}{=} 0, \quad \left(q-2 \right) \tau_{nn} - \tau \overset{\Sigma}{=} 0.
\end{equation}
Moreover, taking a derivative of (\ref{mattersig}) and evaluating at $\Sigma$ reads
\begin{equation}
\begin{split}
&  \overline{\mathbf{g}}_c{}^b \nabla_b \left[ n_a \left( \tau - (q-2) \tau_{nn} \right) \right] + \left( q-2 \right) \overline{\nabla}_c \tau^{\top}_{ a n  }  + \left( q-2 \right) n_a  \tau^{\top}_{ n b } K_c{}^ b  \overset{\Sigma}{=} 0 , \\
& \nabla_n \left[ \left(q-2 \right) \tau_{an} +  n_a \tau \right] -  \nabla_b \tau_a{}^b \overset{\Sigma}{=} 0.
\end{split}
\end{equation}
The first equation is trivially satisfied due to (\ref{mattc1}). The second equation can be decomposed into transversal and intrinsic components,
\begin{equation} \label{mattc2}
\begin{split}
& \nabla_n \left[ \left(q-1 \right) \tau_{nn} - \tau \right] + \tau_{ab} \mathring{K}^{ab}  - \overline{\nabla}{}^a \tau^{\top}_{na} - \frac{H}{3} \left( \tau - 2 \tau_{nn }\right) \overset{\Sigma}{=} 0, \\
&\overline{\nabla}{}^b \tau^{\top}_{ab} - H \tau^{\top}_{an} - \tau^{\top}_{b n} \mathring{K}_{a}{}^b  - \left(q-1 \right) \nabla_n  \tau^{\top}_{an}  \overset{\Sigma}{=} 0,
\end{split}
\end{equation} 
where (\ref{mattc1}) has been used.

\begin{remark}
In \cite{friedrich} a four-dimensional asymptotically de Sitter spacetime $(\widetilde{M}, \widetilde{g}_{ab})$ with $\widetilde{g}_{ab} = \Omega^{-2} \mathtt{g}_{ab}$ (cf. Definition \ref{defadS}) and the stress-energy tensor of a scalar fluid $\phi$ of mass $m$ is considered, i.e.
\begin{equation} \label{strexa}
\widetilde{T}_{ab} = \widetilde{\nabla}_a \phi \widetilde{\nabla}_b \phi - \widetilde{g}_{ab} \left[\frac{1}{2} \left(  \widetilde{\nabla}_c \phi \widetilde{\nabla}^c \phi + m^2 \phi^2 \right) +V \left(\phi \right) \right]
\end{equation}
with $V\left( \phi \right) = \mu \phi^3 + \phi^4 U\left( \phi \right)$, where $\mu$ is a constant and $V' \left( 0 \right)=0$. In terms of a new variable $\psi = \Omega^{-1} \phi$ the stress-energy tensor (\ref{strexa}) can be written as
\begin{equation}
\widetilde{T}_{ab} = \psi^2 \left[ \frac{\Lambda}{3} n_a n_b + \frac{1}{2} \mathtt{g}_{ab} \left( \frac{\Lambda}{3} - m^2 \right) \right] + \mathcal{O} \left( \Omega \right)
\end{equation}
where $n_a$ is the unit normal vector of a conformal boundary. This corresponds to (\ref{tphysunphys}) with $q=0$, and the constraint (\ref{mattc1}) in this case reduces to
\begin{equation}
m^2 = \frac{2}{3} \Lambda,
\end{equation} 
which matches the condition for the regularity of the conformal field equations derived there.
\end{remark}

\subsection*{Conformal fundamental forms}
\hfill

The starting point in the construction of conformal fundamental forms is at the second jet of the scale $\widetilde{\sigma}$ (the almost-Einstein-matter equation operator). The key object will be denoted by
$E_{ab}$ and is given by 
\begin{equation} \label{kemat}
\begin{split}
E_{ab} &  := p^* \left(  \hat{D}^A I^B_{\widetilde{\sigma}}  \right) =   Z_{Bb} \nabla_a I^B_{\widetilde{\sigma}} \\
& = \nabla_a \nabla_b \widetilde{\sigma} + \widetilde{\sigma} P_{ab} - \frac{1}{4} \mathbf{g}_{ab} \left( \Delta \widetilde{\sigma} + \widetilde{\sigma} J \right).
\end{split}
\end{equation}
Indeed, this tensor field, which is smooth to the boundary, restricts to ($\sqrt{\lambda}$ times) the second
fundamental form there, but on the interior gives the traceless part
of the Schouten for $\widetilde{g}_{ab}$. Due to (\ref{aleinm}) we
know that $E_{ab}$ can be associated with the stress-energy tensor
density in the following way,
\begin{equation} \label{kemat2}
E_{ab} =  \frac{\widetilde{\sigma}^{q+1}}{2} \mathring{\tau}_{ab}.
\end{equation}
This observation will allow us to relate the conformal fundamental
forms of $\Sigma$ to the matter fields on the conformal boundary. The
immediate consequence of (\ref{kemat2}) and (\ref{dern2}) is the fact
that $\Sigma$ is an umbilic hypersurface for asymptotically de Sitter
spacetimes ($q \geq 0$). More generally, the almost-Einstein-matter
scale $\widetilde{\sigma}$ has the following properties:
\begin{itemize}
\item $\nabla_n^k E_{ab} \overset{\Sigma} = 0 $ for $k \leq q$;
\item $(q+3)^{\rm rd}$ jet of $\widetilde{\sigma}$ on $\Sigma$ will involve the trace-free stress-energy tensor density $\tau_{ab}$ there;
\item  $(q+4)^{\rm th}$ and higher jets of $\widetilde{\sigma}$ on $\Sigma$ will involve (at least first) derivatives of the trace-free stress-energy tensor density $\tau_{ab}$ there;
\end{itemize}

Before moving forward, let us recall (from (\ref{i2}) and (\ref{deri})) formulas for the almost Einstein scale tractor and its derivative in four dimensions. We have
 \begin{equation} \label{normi4}
I_{\widetilde{\sigma}}^2   = -  \lambda+ \frac{1 }{12 } \widetilde{\sigma}^{q+2} \tau 
 \end{equation}
 and
\begin{equation} \label{derimat}
\nabla_a I_{\widetilde{\sigma}}  = \frac{\widetilde{\sigma}^{q}}{ 6 }   \left(\begin{array}{c}
0 \\
3 \widetilde{\sigma}  \mathring{\tau}_{ab} \\
- \left(q+1\right)  \mathring{\tau}_a{}^b \nabla_b \widetilde{\sigma}  - \widetilde{\sigma} \nabla_b \mathring{\tau}_a{}^b
\end{array}\right) .
\end{equation} 

In the sequel we will focus on $q=0,1,2$. This is motivated by the
fact that those values are most commonly used in the analysis of the
conformal extensions of asymptotically de Sitter spacetimes (see
e.g. \cite{javk, a2, senovilla, friedrich}). Moreover, this
choice will allow us to focus on the first five conformal fundamental
forms of $\Sigma$, which turn out to consist of projections of Weyl, Cotton and
Bach tensors (plus the induced metric and extrinsic curvature), well-known objects in the conformal geometry.
 
Following \cite{fundforms}, we will present a construction of
conformal fundamental forms. The basic principle behind it is to
consider normal derivatives of $E_{ab}$ adjusted in a way that makes
the whole expression conformally invariant. To achieve this goal, we
will consider two differential operators, the tractor Robin operator
$\delta_R$ from (\ref{delrob}) and the canonical degenerate Laplacian
$\mathrm{ID}$, defined in (\ref{deglap}), adjusted to act on
trace-free tensorial densities. The former can be defined as
\begin{equation} \label{deltardef}
\delta_R t_{ab}  := \bar{p}^{*} \circ \bar{r} \circ \mathring{\top} \circ \delta_R \circ p \left( t_{ab} \right)
\end{equation}
for $w \neq 3$, where the operators $p$, $p^*$ and $r$ are defined in Section \ref{secprel}. It can be verified that
\begin{equation}
\begin{split}
\delta_R t_{ab}  = & \overline{\mathbf{g}}_a{}^c \overline{\mathbf{g}}_b{}^d \left( \nabla_n t_{cd}   - \frac{w-2}{3} H t_{cd}  \right)  - \frac{1}{3} \overline{\mathbf{g}}_{ab}  \left( n^c n^d \nabla_n t_{cd} - \frac{w-2}{3} H t_{nn} \right) \\
&+ \frac{2}{w-3} \left( \overline{\mathbf{g}}_{(a}{}^c \overline{\nabla}_{b)} t_{n c} - \frac{1}{3} \overline{\mathbf{g}}_{ab} \overline{\nabla}_c t_{n}{}^c  \right) \\
& =  \mathring{\top} \left[  \nabla_n t_{ab}   - \frac{w-2}{3} H t_{ab} + \frac{2}{w-3}  \overline{\nabla}_{(a} t_{b)n}^{\top} \right],
\end{split}
\end{equation}
such that
\begin{equation}
\delta_R:  \mathcal{E}_{(ab)_0} [w] \to \overline{\mathcal{E}}_{(ab)_0} [w-1] \quad \mathrm{for} \quad w \neq 3,
\end{equation}
i.e.\ $\delta_R t_{ab} $ takes values in a weight twisting of the trace-free part  part  of the symmetric covariant submanifold 2-tensors. The canonical degenerate Laplacian $\mathrm{ID}$ acting on the tensorial density $t_{ab}$ can be defined as
\begin{equation}
\mathrm{ID} t_{ab} := p^{*} \circ r \circ I_{\widetilde{\sigma}}^A D_A \circ p \left( t_{ab} \right),
\end{equation}
or
\begin{equation} \label{deglap2}
\begin{split}
\mathrm{ID} t_{ab} = &2  (w-1) \sqrt{\lambda} \bigg(\left[  \nabla_n + (w-2) \rho \right] t_{ab} - \tfrac{2(w-2)}{(w-3)(w+2)} n_{(a} \nabla \cdot t_{b)_0} \\ & + \frac{2}{w-3} \left[n \cdot \nabla_{(a} t_{b)_0} + t_{(a } \cdot   \nabla_{b)_0} n  \right] \bigg)  \\
&- \widetilde{\sigma} \bigg( \Delta t_{ab} + (w-2) J t_{ab} + \tfrac{8}{(w-3)(w+2)} \nabla_{(a} \nabla \cdot  t_{b)_0}  - 4 P_{(a} \cdot t_{b)_0}  \bigg),
\end{split}
\end{equation}
hence
\begin{equation}
\mathrm{ID}: \mathcal{E}_{(ab)_0} [w] \to \mathcal{E}_{(ab)_0} [w-1] \quad \mathrm{for} \quad w \neq -2, 3.
\end{equation}

The formula for conformal fundamental forms can now be given in terms of applying appropriate power of the degenerate Laplacian $\mathrm{ID}$ and $\delta_R$ to $E_{ab}$. 

\begin{definition} \label{fundformdef}
Let $i \in \{ 0, 1, 2,3,4\}$. A conformal fundamental form $\mathring{\mathcal{K}}^{(i+2)}_{ab}$ can be defined as
\begin{equation}
\mathring{\mathcal{K}}^{(i+2)}_{ab} :=
\begin{cases}
\mathring{K}_{ab} \quad \mathrm{for} \quad i=0,\\
\delta_R \circ \left( \mathrm{ID} \right)^{i-1}  E_{ab} \quad  \mathrm{for} \quad 1 \leq i \leq 4.
\end{cases}
\end{equation}
\end{definition}

The upper bound for $i$ in Definition \ref{fundformdef} is dictated by the fact that the conformal weight of $\mathrm{ID}^3 E_{ab}$ is $-2$, so the $\mathrm{ID}$ operator applied to this quantity will have a pole (because of the coefficients $\frac{1}{w+2}$).

Unlike in the current scenario, the definition of conformal fundamental forms from \cite{fundforms} relied on the fact that $I^2_{\widetilde{\sigma}} = \pm 1 + \mathcal{O} \left(\widetilde{\sigma}^4 \right) $ (which can always be achieved after improving the scale, see e.g. \cite[Theorem 1.3]{GW}). In order to retain the same formulae for those objects, we will restrict ourselves to the discussion of the conformal fundamental forms up to $\mathring{\mathcal{K}}^{(q+4)}_{ab}$ for the given value of the decay parameter $q$. It can be verified that the higher-order tensors include contributions from the trace $\tau$. The conformal fundamental forms up to $\mathring{\mathcal{K}}^{(q+4)}_{ab}$ are given solely in terms of geometric quantities because of the fact that Definition \ref{fundformdef} involves taking a \emph{projected} trace-free part (via $\delta_R$) of a differential operator acting on $E_{ab}$. It also implies that $\mathring{\mathcal{K}}^{(6)}_{ab}$ will contain derivatives of the Bach tensor. An object of this transverse order is not usually considered in the literature, so $\mathring{\mathcal{K}}^{(6)}_{ab}$ will be excluded from the discussion presented here.

In the sequel we will use Definition \ref{fundformdef} together with relation (\ref{kemat2}) to derive constraints relating conformal
fundamental forms of $\Sigma$ and the matter fields there. The first case has already been discussed above -- we observed that the second fundamental form $\mathring{K}_{ab}$ vanishes when $q \geq 0$, i.e. $\Sigma$ is an umbilic hypersurface. It should be noted that by Definition \ref{fundformdef}, the higher conformal fundamental forms will also be trace-free.

\subsection*{Third fundamental form - the Weyl tensor}
\hfill

To obtain the third jet of $\widetilde{\sigma}$ one needs to apply $\delta_R$ to $E_{ab}$, i.e.
\begin{equation}
\mathring{\mathcal{K}}^{(3)}_{ab}  := \delta_R \left( \sqrt{\lambda} \nabla_a n_b  + \widetilde{\sigma} P_{ab} - \frac{1}{4} \mathbf{g}_{ab} \left(  \sqrt{\lambda} \nabla_a n^a + \widetilde{\sigma} J \right) \right).
\end{equation}
We will do it in steps. Firstly, we have
\begin{equation}
 E^{\top}_{an} \overset{\Sigma}{=}  \sqrt{\lambda} \overline{\mathbf{g}}_a{}^c  \nabla_n n_c \overset{\Sigma}{=} 0 ,
\end{equation}
from (\ref{nablan}).
So the formula for $\mathring{\mathcal{K}}^{(3)}_{ab}$  reduces to
\begin{equation}
\mathring{\mathcal{K}}^{(3)}_{ab} :=  \mathring{\top} \left[  \nabla_n  E_{ab}   + \frac{H}{3}  E_{ab} \right] =  \mathring{\top}   \left( \nabla_n  E_{ab} \right) +  \frac{H}{3} \sqrt{\lambda} \mathring{K}_{ab},
\end{equation}
where $\mathring{K}_{ab}$ will be set to zero later on ($\Sigma$ is
umbilic). Moreover, we can directly use the definition (\ref{riemdef})
of the Riemann tensor and its decomposition into the Weyl and Schouten
tensors (\ref{riemdec}) to compute the normal derivative of $E_{ab}$,
\begin{equation}
\begin{split}
\nabla_n E_{ab} &  \overset{\Sigma}{=}  \sqrt{\lambda} \bigg[ C_{nabn} + 2 n_{(a} P_{b)n}  - \nabla_c n_b \nabla_a n^c + \nabla_a \nabla_n n_b \\
& -\frac{1}{4} \mathbf{g}_{ab} \left( 2 P_{nn}  -  \nabla_a n_b \nabla^b n^a + \nabla^a \nabla_n n_a \right) \bigg] \bigg|_{\Sigma} .
\end{split}
\end{equation}
Hence,
\begin{equation}
 \mathring{\top}   \left( \nabla_n  E_{ab} \right)  \overset{\Sigma}{=}  \sqrt{\lambda} \left( C^{\top}_{nabn}  - K_{bc} K_{a}{}^c + \frac{H}{3} \mathring{K}_{ab} +\frac{1}{3} \overline{g}_{ab} K_{cd} K^{cd} \right).
\end{equation}
Thus we ultimately obtain the following,
\begin{equation} \label{3}
\mathring{\mathcal{K}}^{(3)}_{ab}: =  \sqrt{\lambda} \left( C^{\top}_{nabn}  - \mathring{K}_{bc} \mathring{K}_{a}{}^c + \frac{1}{3} \overline{g}_{ab} \mathring{K}_{cd} \mathring{K}^{cd} \right),
\end{equation} 
where the r.h.s. of this expression (modulo constant) is called the Fialkow tensor \cite{juhl}. It reduces to 
\begin{equation} \label{k3weyl}
\mathring{\mathcal{K}}^{(3)}_{ab} :=  \sqrt{\lambda}  C^{\top}_{nabn}  
\end{equation}
on totally umbilic hypersurfaces (meaning $\mathring{K}_{ab}
\overset{\Sigma}{=} 0$). Therefore, the third conformal
fundamental form of a conformal boundary of asymptotically de Sitter
spacetime is proportional to the electric part of the Weyl tensor.

To relate this to matter fields on $\Sigma$ one needs to apply
$\delta_R$ to the right-hand side of (\ref{kemat2}), i.e.
\begin{equation} \label{k3mat}
\delta_R \left( \frac{1}{2} \widetilde{\sigma}^{q+1} \mathring{\tau}_{ab} \right) = - \frac{\sqrt{\lambda}}{2} \left(q+1 \right) \widetilde{\sigma}^q  \mathring{\top} \left( \tau_{ab} \right) \bigg|_{\Sigma}.
\end{equation}
The combination of (\ref{k3weyl}) and (\ref{k3mat}) yields
\begin{equation} \label{fialkowmat}
 C_{na n b} \overset{\Sigma}{=}  \frac{1}{2} \left(q+1 \right) \widetilde{\sigma}^q  \mathring{\top} \left( \tau_{ab} \right) \bigg|_{\Sigma},
\end{equation}
which gives a non-trivial relation between the third conformal fundamental form of $\Sigma$ and the stress-energy tensor when $q=0$:
\begin{equation}
 C_{na n b} \overset{\Sigma}{=}  \frac{1}{2} \mathring{\top} \left( \tau_{ab} \right) \quad \mathrm{for} \quad q=0.
\end{equation}
For higher $q$ we conclude that the third fundamental form
$C_{na n b} $ must be zero along $\Sigma$. 

\subsection*{Fourth fundamental form - the Cotton tensor} 
\hfill

The action of the canonical degenerate Laplacian $\mathrm{ID}$ (\ref{deglap2}) on $E_{ab}$ simplifies to
\begin{equation}
\mathrm{ID} E_{ab} = - \widetilde{\sigma} \bigg( \Delta E_{ab} - J E_{ab} - \tfrac{4}{3} \nabla_{(a} \nabla \cdot E_{b)_0}  - 4 P_{(a} \cdot E_{b)_0}  \bigg).
\end{equation}
The fourth trace-free fundamental form then reads,
\begin{equation} 
\begin{split}
\mathring{\mathcal{K}}^{(4)}_{ab} := & \delta_R \circ \mathrm{ID} E_{ab}  = \sqrt{\lambda }  \mathring{\top}   \bigg( \Delta E_{ab} - J E_{ab} - \tfrac{4}{3} \nabla_{(a} \nabla \cdot E_{b)_0}  - 4 P_{(a} \cdot E_{b)_0}  \bigg)  \\ 
& = -2 \lambda  \mathring{\top}  \left(  \mathring{K}^{cd} C_{acbd } \right),
\end{split}
\end{equation}
where (\ref{kemat}), the definition (\ref{riemdef}) of the Riemann tensor $R_{abcd}$ in terms of a commutator of covariant derivatives acting on $n^a$, together with its decomposition (\ref{riemdec}) and the Bianchi identity (\ref{bianchi}) have been used. The third equality can be obtained by noticing that the principal parts (the highest-order derivatives acting on $n^a$) cancel each other out. Hence, $ \mathring{\mathcal{K}}^{(4)}_{ab}$ vanishes if $\mathring{K}_{ab} \overset{\Sigma}{=}0$ and the constraints on the
behaviour of the matter fields on the conformal boundary $\Sigma$ for asymptotically de Sitter spacetimes can be obtained by applying $\delta_R \circ \mathrm{ID}$ to the stress-energy counterpart of $E_{ab} $ from (\ref{kemat2}), as follows. If $q=0$, then
\begin{equation} \label{ccot0}
\nabla_n \tau^{\top}_{ab} + \frac{2}{3} H \tau^{\top}_{ab} \overset{\Sigma}{=}0,
\end{equation}
where (\ref{mattc1}) and the normal component of (\ref{mattc2}) with $\mathring{K}_{ab} \overset{\Sigma}{=}0$ have been used, i.e. 
\begin{equation} \label{constr0}
\tau^{\top}_{a n} \overset{\Sigma}{=}0, \quad \tau \overset{\Sigma}{=} - 2 \tau_{nn}, \quad \nabla_n \left(\tau + \tau_{nn} \right) \overset{\Sigma}{=} \frac{4}{3} H \tau_{nn}.
\end{equation}
In the case where $q=1$ a simple condition 
\begin{equation} \label{matcon}
 \mathring{\top} \left(  \tau_{ab} \right) \overset{\Sigma}{=} 0
\end{equation}
arises. After combining it with the constraint (\ref{mattc1}), we get
\begin{equation} \label{q1tau}
\tau_{ab} \overset{\Sigma}{=} - n_a n_b \tau,
\end{equation}
which was also derived in the context of conformal Einstein field equations, see e.g. \cite{a2, senovilla}.

For $q=2$ the constraint is trivial because the decay of matter fields
is too fast to be captured by $\delta_R \circ \mathrm{ID} $. Thus our
definition of $\mathring{\mathcal{K}}^{(4)}_{ab}$, as above, fails to give a
meaningful constraint relating the geometry of $\Sigma$ to the
stress-energy tensor for asymptotically de Sitter spacetimes in this
case. However, based on the results from \cite{fundforms} we will make
use of the following choice for the fourth conformal fundamental form
of $\Sigma$, sometimes called (modulo constant) a hypersurface Bach
tensor.
\begin{definition} \label{def_ff4}
The fourth conformal fundamental form $\mathring{\mathcal{K}}^{(4)}_{ab}$ can be defined as
\begin{equation} \label{fundform4def}
\mathring{\mathcal{K}}^{(4)}_{ab} := \lambda \left( \overline{\nabla}{}^c C^{\top}_{n(ab)c} - A^{\top}_{(a |n|b)} - \frac{1}{3} H C_{nanb} \right).
\end{equation}
\end{definition}
It can be checked that (\ref{fundform4def}) has the desired properties, i.e. is trace-free, intrinsic to $\Sigma$ and has conformal weight $-1$. The latter property can be verified with the use of the following conformal transformation rules ($\widetilde{g}_{ab} = f^2 g_{ab}$),
\begin{equation}
\begin{split}
\widetilde{A}^{\top}_{(a |\widetilde{n}|b)} &  = \frac{1}{f} \left( A^{\top}_{(a |n|b)} - n_d \Upsilon^d C_{nanb} - \overline{\Upsilon}^c C^{\top}_{n (a b) c } \right), \\
\frac{1}{3} \widetilde{H} C_{\widetilde{n}a \widetilde{n} b} & = \frac{1}{f}\left( \frac{H}{3} + n_d \Upsilon ^d \right) C_{nanb}, \\
 \widetilde{\overline{\nabla}{}^c} \widetilde{C}^{\top}_{\widetilde{n} abc} & = \frac{1}{f} \left( \overline{\nabla}{}^c C^{\top}_{nabc} - \overline{\Upsilon}^c   C^{\top}_{nbac} \right) .
\end{split}
\end{equation}

The choice for the fourth conformal fundamental form from Definition \ref{def_ff4} has been motivated by considering the action of $ \delta_R \circ \mathrm{ID} $ on $ E_{ab}$ computed in arbitrary dimension $\dim M = n$. With the assumption that $\Sigma$ is an umbilic hypersurface ($\mathring{K}_{ab} \overset{\Sigma}{=} 0$, i.e. $q \geq 0$), one obtains the following,
\begin{equation} \label{iv}
 \frac{\delta_R \circ \mathrm{ID} E_{ab} }{n-4}   = \lambda  \cancelto{0}{\overline{\nabla}{}^c C^{\top}_{nabc}} + \lambda  \left(n-5 \right) \left(  A^{\top}_{a n b } + \frac{H}{n-1} C^{\top}_{n a n b} \right),
 \end{equation}
where the $\overline{\nabla}{}^c C^{\top}_{nabc}$ term vanishes
because $\mathring{K}_{ab} \overset{\Sigma}{=} 0$ and the Codazzi
equation (\ref{codazzi}). It should be noted that the fourth conformal
fundamental form Definition \ref{def_ff4} does not, in dimension four,
arise the construction which simply factors through the jets of
$E_{ab}$. This is evident from our calculations here. Conceptually
this shows that (and is happening because) in dimension four it is
an image of the Dirichlet-to-Neumann map for the conformal Einstein
field equations -- cf. our discussion in the Introduction and
\cite{GrahamDN}. 

Since we assume that the Weyl tensor $C_{abcd}$ is smooth in the neighbourhood of the conformal boundary $\Sigma$, and is seen to
vanish along $\Sigma$ for $q \geq 1$, the hypersurface divergence of $\mathring{\mathcal{K}}^{(4)}_{ab}$ from Definition \ref{def_ff4} can be determined by the stress-energy tensor on $\Sigma$ in that case.

\begin{theorem}
Let $(\widetilde{M}, \widetilde{g}_{ab} )$ be a four-dimensional asymptotically de Sitter spacetime with the stress-energy tensor $\widetilde{T}_{ab}$ and the conformal extension $(M, \mathtt{g}_{ab} )$ such that
\begin{equation} 
\widetilde{T}_{ab} = \widetilde{\si}^q \tau_{ab}, \quad q = 1,2.
\end{equation}
Then, the Weyl tensor and $\tau^{\top}_{a n}$ vanish on the conformal boundary and 
\begin{equation}
\overline{\nabla}^b \mathring{\mathcal{K}}^{(4)}_{ab} \overset{\Sigma}{=}  
\begin{cases}
 -\lambda^2 j_a  + \frac{\lambda^{3/2}}{3} \overline{\nabla}_a \tau \quad \mathrm{for} \quad q=1, \\
 - \lambda^2 \tau^{\top}_{a n}   \quad \mathrm{for} \quad q=2,
\end{cases}
\end{equation}
where $j_a$ is defined via the expansion $\tau^{\top}_{a n} = \widetilde{\sigma} j_a + \mathcal{O} \left( \widetilde{\sigma} \right)$ and $\mathring{\mathcal{K}}^{(4)}_{ab} = - \lambda A^{\top}_{anb}$.
\end{theorem}
\begin{proof}
We can use (\ref{trc1}) and (\ref{trc2}) to obtain the following,
\begin{equation} \label{peeling}
\begin{split}
   Z_{Dc} [\nabla_a, \nabla_b] I^D_{\widetilde{\sigma}} & =   \widetilde{\sigma} A_{cab} + \sqrt{\lambda} C_{abcn} \\
  & =  \sqrt{\lambda} \left(q+1 \right) \widetilde{\sigma}^q \left( n_{[a} \tau_{b]c} + \frac{1}{3} \tau_{n [a} \mathbf{g}_{b]c} - \frac{1}{3} \tau n_{[a} \mathbf{g}_{b]c}\right) \\ 
  &  + \widetilde{\sigma}^{q+1} \left( \nabla_{[a} \tau_{b]c} - \frac{1}{3} \mathbf{g}_{c [a} \nabla^d \tau_{b] d} + \frac{1}{3} \mathbf{g}_{c [a} \nabla_{b]} \tau\right),
\end{split}
\end{equation}
where the second equality comes from considering the matter counterpart of $\nabla_a I_{\widetilde{\sigma}}$, i.e. the right-hand side of (\ref{derimat}).

It is known that any timelike vector $n^a$ induces a decomposition of the Weyl tensor into its electric and magnetic parts, $C_{a n b n }$ and $C^{\top}_{n abc}$ respectively, which fully determine this tensor in four dimensions. Due to the assumption on $q$, the hypersurface $\Sigma$ is umbilic. In that case the Codazzi equation (\ref{codazzi}) and the constraint (\ref{fialkowmat}) imply $C^{\top}_{n abc} \overset{\Sigma}{=}0$ and $C_{a n b n } \overset{\Sigma}{=} 0$. Hence, the whole Weyl tensor vanishes on the conformal boundary, i.e. $C_{abcd} \overset{\Sigma}{=} 0$ for $q \geq 1$. Because $\Sigma$ is smooth, we can define a rescaled Weyl tensor $K_{abcd}$,
\begin{equation}
K_{abcd} := \frac{1}{\widetilde{\sigma}} C_{abcd},
\end{equation}
which is regular on $M$.
Equation (\ref{peeling}) now implies
\begin{equation} \label{arescc}
A^{\top}_{a n b } \overset{\Sigma}{=} \sqrt{\lambda} K^{\top}_{ a n  b n},
\end{equation}
where constraints (\ref{mattc1}) and (\ref{q1tau}) have been used.

To show the hypersurface divergence constraint, we can use the Bianchi identity and (\ref{peeling}) again to get
\begin{equation}
\begin{split}
\nabla^d K_{d a b c}  & =  \nabla^d \left( \frac{1}{\widetilde{\sigma}} C_{dabc} \right) = \frac{1}{\widetilde{\sigma}} A_{abc} - \frac{\sqrt{\lambda}}{\widetilde{\sigma}^2} C_{n a b c} \\
& =    \sqrt{\lambda} \left(q+1 \right) \widetilde{\sigma}^{q-2} \left( n_{[b} \tau_{c]a} + \frac{1}{3} \tau_{n [b} \mathbf{g}_{c]a} - \frac{1}{3} \tau n_{[b} \mathbf{g}_{c]a}\right) \\
&  +  \widetilde{\sigma}^{q-1} \left( \nabla_{[b} \tau_{c]a} - \frac{1}{3} \mathbf{g}_{a [b} \nabla^d \tau_{c] d} + \frac{1}{3} \mathbf{g}_{a [b} \nabla_{c]} \tau \right)  ,
\end{split}
\end{equation}
which, after contracting with $n^a n^c$ and evaluating on the conformal boundary gives
\begin{equation} \label{rescw}
\overline{\nabla}^b K^{\top}_{a n b n}  \overset{\Sigma}{=}
\begin{cases}
  \frac{2}{3} \sqrt{\lambda }\lim \displaylimits_{\widetilde{\sigma} \to 0 }  \left( \frac{\tau^{\top}_{an}}{\widetilde{\sigma}} \right)  - \frac{1}{3} \overline{\nabla}_a \tau   - \frac{1}{3} \nabla_n \tau^{\top}_{a n } \quad \mathrm{for} \quad q=1, \\ 
  \sqrt{\lambda} \tau^{\top}_{a n}  \quad \mathrm{for} \quad q=2,
\end{cases}
\end{equation}
where constraints (\ref{mattc1}) and (\ref{q1tau}) have been used again. If $q=1$, then $\tau^{\top}_{an} \overset{\Sigma}{=}0$ and
\begin{equation}
\tau^{\top}_{an} = \widetilde{\sigma} j_a + \mathcal{O} \left( \widetilde{\sigma}^2 \right),
\end{equation}
so the ultimate form of (\ref{rescw}) is
\begin{equation}
\overline{\nabla}^b A^{\top}_{a n b }  \overset{\Sigma}{=}  
\begin{cases}
 \lambda j_a  - \frac{\sqrt{\lambda }}{3} \overline{\nabla}_a \tau \quad \mathrm{for} \quad q=1, \\
 \lambda \tau^{\top}_{a n}   \quad \mathrm{for} \quad q=2,
\end{cases}
\end{equation}
where (\ref{arescc}) has been used.
\end{proof} 

\subsection*{Fifth fundamental form - the Bach tensor}
\hfill

Following Definition \ref{fundformdef} and the discussion afterwards, we can now compute the fifth conformal fundamental form of $\Sigma$ for $q \geq 1$. We have
\begin{equation} \label{bachcon}
\mathring{\mathcal{K}}^{(5)}_{ab} := \delta_R \circ \mathrm{ID}^2 E_{ab} = 6 \lambda^\frac{3}{2} \mathring{\top} \left( B_{ab} \right) \quad \mathrm{for} \quad q \geq 1,
\end{equation}
where the definition of the Riemann tensor (\ref{riemdef}) in the context of a commutator of covariant derivatives acting on the normal vector $n^a$ and the definition of the Bach tensor from (\ref{bachcotdef}) have been used. As in the case of the $\mathring{\mathcal{K}}^{(4)}_{ab}$, an important step in deriving the last equality is noticing that the highest-order derivatives acting on $n^a$ cancel each other out.

We will analyze the matter counterpart of $\mathring{\mathcal{K}}^{(5)}_{ab}$ for two different values of the parameter $q$. If $q=1$, then \eqref{kemat} implies
\begin{equation}
 \delta_R \circ \mathrm{ID}^2 E_{ab}  = 2 \lambda^2 \left( 9 \nabla_n \tau^{\top}_{ab} - 3 \overline{\mathbf{g}}_{ab} \nabla_n \tau_{nn} + 8 H \tau^{\top}_{ab} - H \overline{\mathbf{g}}_{ab} \tau_{nn} \right),
\end{equation}
where (\ref{mattc1}) and (\ref{mattc2}) have been used. Therefore,
\begin{equation}
\mathring{\mathcal{K}}^{(5)}_{ab}   \overset{\Sigma}{=} 2 \lambda^2 \left( 9 \nabla_n \tau^{\top}_{ab} - 3 \overline{\mathbf{g}}_{ab} \nabla_n \tau_{nn} + 8 H \tau^{\top}_{ab} - H \overline{\mathbf{g}}_{ab} \tau_{nn} \right)\quad \mathrm{for} \quad q=1.
\end{equation}

For $q=2$,  $ \delta_R \circ \mathrm{ID}^2 $ applied to the right-hand side of (\ref{kemat}) yields
\begin{equation}
 \delta_R \circ \mathrm{ID}^2 E_{ab} =  \delta_R \circ \mathrm{ID}^2 \left(\frac{1}{2} \widetilde{\sigma}^3 \mathring{\tau}_{ab} \right) \overset{\Sigma}{=} - 18 \lambda^\frac{5}{2}  \mathring{\top} \left( \tau_{ab} \right),
\end{equation}
where the constraint (\ref{mattc1}) has been taken into account ($\tau  \overset{\Sigma}{=}0$). Ultimately,
\begin{equation}
\mathring{\mathcal{K}}^{(5)}_{ab}   \overset{\Sigma}{=}  - 18 \lambda^\frac{5}{2}  \mathring{\top} \left( \tau_{ab} \right)  \quad \mathrm{for} \quad q=2.
\end{equation}

The constraints derived above allowed us to relate the projected part of the Bach tensor with the matter fields on the conformal boundary $\Sigma$. In order to associate the other components of $B_{ab}$ with the stress-energy tensor on $\Sigma$ a commutator of the two Thomas-D operators can be used. The strategy is to apply it to the scale tractor $I_{\widetilde{\sigma}}$ and compute the resulting expression in two different ways, either by exploiting the form of $\nabla_a  I_{\widetilde{\sigma}}$ given in terms of geometric objects or by using (\ref{derimat}). Firstly, we have
\begin{equation} \label{bach1}
 D_{[A} D_{B]}  \left( I_{\widetilde{\sigma}} \right)_C \overset{\Sigma}{=}  2 \sqrt{\lambda} X_C X_{[A} Z_{B]}{}^c B_{c n},
\end{equation}
which can be obtained as a direct application of the definition of the Thomas-D operator from Section \ref{secprel}. On the other hand, 
\begin{equation} \label{bach2}
D_{[A} D_{B]}  I_{\widetilde{\sigma}} = - X_{[A} \left(\Delta - J  \right) \left( 2 Z_{B]}{}^a  \nabla_a  I_{\widetilde{\sigma}}   - X_{B]} \nabla^b  \nabla_b  I_{\widetilde{\sigma}} \right),
\end{equation}
so after using (\ref{derimat}) this expression can be written solely in terms of the stress-energy tensor density $\tau_{ab}$. After comparing it with (\ref{bach1}) we get the following,
\begin{equation}
B_{nn} \overset{\Sigma}{=} 0, \quad B^{\top}_{a n}  \overset{\Sigma}{=} \sqrt{\lambda} \left( \nabla_n \tau^{\top}_{an} + \frac{1}{3} \overline{\nabla}_a \tau \right)  \quad \mathrm{for} \quad q=1,
\end{equation}
and
\begin{equation}
B_{nn} \overset{\Sigma}{=} 0,  \quad B^{\top}_{a n} \overset{\Sigma}{=} - \lambda \tau^{\top}_{an} \quad \mathrm{for} \quad q=2,
\end{equation}
where constraints (\ref{mattc1}) and (\ref{mattc2}) have been used.
\subsection*{Summary} 
\hfill

The results of this section can be stated in the form of the following theorem.
\begin{theorem}
Let $(\widetilde{M}, \widetilde{g}_{ab} )$ be a four-dimensional asymptotically de Sitter spacetime with the stress-energy tensor density $\tau_{ab}$ of weight $q \in \{ 0,1,2 \}$. Then the conformal boundary $\Sigma$ is an umbilic hypersurface, and
\begin{equation}\label{temp}
\begin{split}
& \left(q-2 \right) \tau_{an} +  n_a \tau \overset{\Sigma}{=} 0, \\
 & \nabla_n \left[ \left(q-2 \right) \tau_{an} +  n_a \tau \right] -  \nabla_b \tau_a{}^b \overset{\Sigma}{=} 0,
\end{split}
\end{equation}
where
\begin{equation}
n_a n^a = -1 + \frac{2}{ n \lambda }  \widetilde{\sigma} \left( \Delta + J  \right) \widetilde{\sigma} + \mathcal{O} \left( \widetilde{\sigma}^{q+2}\right)  \overset{\Sigma}{=} -1.
\end{equation}
Moreover, the constraints relating conformal fundamental forms of $\Sigma$ to the matter fields there have the following form for the specific values of $q$:
\begin{itemize}
\item $q=0$
\begin{equation}
\begin{split}
 & C_{na n b} \overset{\Sigma}{=}  \frac{1}{2} \mathring{\top} \left( \tau_{ab} \right), \\
 & \nabla_n \tau^{\top}_{ab} + \frac{2}{3} H \tau^{\top}_{ab} \overset{\Sigma}{=}0,
\end{split}
\end{equation}
\item $q=1$
\begin{equation}
\begin{split}
 & C_{abcd} \overset{\Sigma}{=} 0, \quad  \tau_{ab} \overset{\Sigma}{=} - n_a n_b \tau, \\
 & \overline{\nabla}^b A^{\top}_{a n b } \overset{\Sigma}{=}   \lambda j_a  - \frac{\sqrt{\lambda }}{3} \overline{\nabla}_a \tau, \\
 & \mathring{\top} \left(  B_{ab} \right)   \overset{\Sigma}{=} \frac{\sqrt{\lambda}}{3} \left( 9 \nabla_n \tau^{\top}_{ab} - 3 \overline{\mathbf{g}}_{ab} \nabla_n \tau_{nn} + 8 H \tau^{\top}_{ab} - H \overline{\mathbf{g}}_{ab} \tau_{nn} \right), \\
 & B_{nn} \overset{\Sigma}{=} 0, \quad B^{\top}_{a n}  \overset{\Sigma}{=} \sqrt{\lambda} \left( \nabla_n \tau^{\top}_{an} + \frac{1}{3} \overline{\nabla}_a \tau \right),
\end{split}
\end{equation}
\item $q=2$
\begin{equation}
\begin{split}
& C_{abcd} \overset{\Sigma}{=} 0, \quad \overline{\nabla}^b A^{\top}_{a n b }  \overset{\Sigma}{=} \lambda \tau^{\top}_{a n}, \\
 & \mathring{\top} \left( B_{ab}  + 3 \lambda \tau_{ab} \right) \overset{\Sigma}{=} 0, \\
 & B_{nn} \overset{\Sigma}{=} 0,  \quad B^{\top}_{a n} \overset{\Sigma}{=} - \lambda \tau^{\top}_{an}.
\end{split}
\end{equation}
\end{itemize}
\end{theorem}

\section{Conclusions}
We have derived constraints relating the conformal fundamental forms
of the conformal infinity of asymptotically de Sitter spacetime
with its stress-energy tensor. As a result,  projected parts of the
Weyl and Bach tensors and a divergence of the Cotton tensor on the
conformal boundary have been determined locally by the matter fields
there. The natural application of this result is the study of matching
of spacetimes in the Conformal Cyclic Cosmology scenario, where the
conformal infinity of the asymptotically de Sitter spacetime is
identified with the Big Bang hypersurface of the second conformally
extended solution of the Einstein field equations. Mimicking the
procedure of joining spacetimes via the Darmois-Israel junction condition,
the natural strategy to consider in this setting is the identification
of the conformal fundamental forms of the conformal boundaries of
asymptotically de Sitter and initial singularity spacetimes to obtain
constraints on the matter content on the transition hypersurface. This
will be considered elsewhere.

\subsection*{Acknowledgements}
A.R.G.\ gratefully acknowledges support from the Royal Society of New Zealand via Marsden Grants 16-UOA-051 and 19-UOA-008. J.K.\ would like to thank Pawe\l \ Nurowski for the encouragement to pursue the topic of this work. He acknowledges funding received from the Norwegian Financial Mechanism 2014-2021, project registration number UMO-2019/34/H/ST1/00636.

\end{document}